\documentclass{article}
\usepackage{amsmath}
\usepackage{amssymb}
\usepackage{amsthm}

\newcommand{\comment}[1]{}

\newtheorem{Theorem}{Theorem}
\newtheorem{Lemma}{Lemma}
\newtheorem{Corollary}{Corollary}

\title{Optimal bounds on the classical value of the repeated CHSH game}
\author{Andris Ambainis\thanks{Department of Computing, University of Latvia, e-mail: ambainis@lu.lv}}
\date{}

\begin{document}

\maketitle

\begin{abstract}
We show that the maximum winning probability in the repeated CHSH game for classical strategies is at most $\left(\frac{1+\sqrt{5}}{4}\right)^n$. 
Together with a matching lower bound by Barak et al. \cite{B+}, this determines the asymptotic value of the repeated CHSH game exactly. 
We also show that, if an XOR game has a gap between classical and quantum values, there is also a gap between the asymptotic values for 
the repeated game.
\end{abstract}

\section{Introduction}

Non-local games \cite{Cleve} provide a simple and transparent way to demonstrate the distinction between classical and quantum correlations. In a non-local game, two spatially separated players receive correlated questions from a referee and must produce answers without communicating. Their maximum winning probability when they share only classical randomness is called the classical value of the game, whereas allowing the players to share entanglement gives its quantum, or entangled, value. 

The best-known example is the CHSH game, introduced in connection with Bell’s theorem \cite{CHSH69,Cleve}. In one instance of the CHSH game, Alice and Bob receive independent uniformly random bits $a, b$ and have to produce output bits $x, y$, winning if
\[
x\oplus y = a b .
\]
Classically they can win with probability at most $val(CHSH) = 3/4$, while an entangled strategy achieves the strictly larger value
$val^*(\mathrm{CHSH})
=\cos^2(\pi/8)
=\frac{2+\sqrt 2}{4}$.
This separation is one of the simplest manifestations of quantum non-locality.

A basic operation on non-local games is parallel repetition, in which several independent instances of a game are played simultaneously and the players must win all of them. Despite the simplicity of its definition, its effect on the value of a game is very nontrivial, giving rise to a long line of research and substantially different results for different classes of games and strategies ~\cite{Raz98,CSUU08,Hol09,DS,DSV15}.

In the repeated CHSH game $CHSH^{\otimes n}$, the referee samples $n$ independent instances of CHSH, and the players win only if they win every instance. For entangled quantum strategies, XOR games (a class of games which includes CHSH) satisfy perfect parallel repetition \cite{CSUU08}. This means that the winning probability, $val^*(\mathrm{CHSH}^{\otimes n})$ is equal to
$(val^*(\mathrm{CHSH}))^n =\left(\frac{2+\sqrt2}{4}\right)^n$.

In contrast, determining the classical winning probability is much more difficult. Playing an optimal single-copy strategy independently gives $val(G^{\otimes n})\geq val(G)^n$.
The reverse inequality, however, need not hold: each player may correlate their answers across coordinates, even though the questions themselves are independent. This is true already for two copies: although $val(\mathrm{CHSH})=3/4$, two copies can be won classically with probability $val(\mathrm{CHSH}^{\otimes 2})
=\frac{10}{16}
>\left(\frac34\right)^2$ \cite{CSUU08}. For 3 copies, we have $val(\mathrm{CHSH}^{\otimes 2})
= \frac{31}{64}$ \cite{CSUU08,Araujo}.

General parallel-repetition theorems, beginning with Raz’s theorem \cite{Raz98} and its subsequent refinements, imply that the classical value of any nontrivial two-player game decreases exponentially under the repetition. They generally do not, however, determine the correct exponential rate for a particular game. It is therefore natural to consider the asymptotic value (whose limit exists by supermultiplicativity and Fekete's lemma)
$val_\infty(G)
:=\lim_{n\to\infty}val(G^{\otimes n})^{1/n}$.
For CHSH, Barak, et al. \cite{B+} constructed correlated strategies showing that
\[
val_\infty(\mathrm{CHSH})
\geq \cos(\pi/5)
=\frac{1+\sqrt5}{4}
\approx 0.809016.
\]
This is strictly larger than both the  rate $\sqrt[3]{31}/4$ obtained from repeating the optimal three-copy strategy.

Our main result gives the matching upper bound. We prove that, for every $n\geq 1$, we have
$val(\mathrm{CHSH}^{\otimes n})
\leq
\left(\frac{1+\sqrt5}{4}\right)^n$.
Together with the construction of \cite{B+}, this determines the asymptotic classical value exactly:
$val_\infty(\mathrm{CHSH})
=\frac{1+\sqrt5}{4}$.

Therefore, the repeated CHSH game exhibits distinct classical and quantum exponential rates, approximately 0.809016 and 0.853553, respectively.

The proof is short and uses the analytic framework of Dinur and Steurer \cite{DS}. Their framework considers a relaxation 
$\operatorname{val}^{+}(G)$ of the non-local game that upper-bounds the classical value and behaves multiplicatively under tensor products. Consequently,
$
val(G^{\otimes n})
\leq \operatorname{val}^{+}(G)^n.
$
We define a further relaxation, $\operatorname{val}^{++}(G)$ and show that,  for CHSH, symmetry reduces the computation of this relaxation to an elementary optimization over four nonnegative variables, whose optimum is exactly $(1+\sqrt5)/4$. 

We also show that the behaviour in this example is typical in the following sense: if there is a gap 
between the single-copy values $val(G)$ and $val^*(G)$, there is also a gap between the asymptotic value $val_{\infty}(G)$ and 
$val^*(G)$.

The results in this paper were originally obtained in 2014 and reported in \cite{A}. Due to the ongoing interest in the topic, we have decided to publish them.

\section{Setting}

A two-player non-local game \(G\) is specified by finite sets \(A\) and
\(B\) of questions, finite sets \(X\) and \(Y\) of answers, a probability
distribution \(p\) on \(A\times B\), and a family of predicates specifiying the winning conditions
\[
P_{a,b}:X\times Y\longrightarrow\{0,1\}.
\]
The referee samples a pair of questions \((a,b)\) according to \(p\),
sends \(a\) to Alice and \(b\) to Bob, and receives answers \(x\in X\)
and \(y\in Y\), respectively. The players win if
\(P_{a,b}(x,y)=1\).

A deterministic classical strategy is given by functions \(f:A\to X\)
and \(g:B\to Y\). Its winning probability is
\[
\sum_{a\in A,\,b\in B}
p(a,b)P_{a,b}\bigl(f(a),g(b)\bigr).
\]
The classical value of \(G\), denoted by \(\operatorname{val}(G)\), is
the maximum of this expression over all \(f\) and \(g\). Allowing Alice
and Bob to share classical randomness does not increase the value,
since every randomized strategy is a convex combination of
deterministic ones.

In a quantum strategy, Alice and Bob share a bipartite quantum state
\(\rho\), prepared before they receive their questions. For each
\(a\in A\), Alice applies a measurement
\(\{M_x^a\}_{x\in X}\), and for each \(b\in B\), Bob applies a
measurement \(\{N_y^b\}_{y\in Y}\). The entangled value of \(G\),
denoted by \(\operatorname{val}^*(G)\), is the supremum of
\[
\sum_{a,b,x,y}
p(a,b)P_{a,b}(x,y)
\operatorname{Tr}\!\left( \left(
\,M_x^a\otimes N_y^b \right) \rho
\right)
\]
over all finite-dimensional bipartite states and local measurements.

An XOR game is a game with \(X=Y=\{0,1\}\) for which, for every
\((a,b)\) in the support of \(p\), the winning condition is such that\begin{equation}
\label{eq:xor}
P_{a,b}(x,y)=1
\quad\Longleftrightarrow\quad
x\oplus y=s(a,b)
\end{equation}
for some  \(s(a,b)\in\{0, 1\}\). CHSH is an example of a XOR game: \(A=B=X=Y=\{0,1\}\), the distribution \(p\) is uniform,
and \(s(a,b)=ab\).

A projection game is a game in which, for every \((a,b)\) in the support
of \(p\) and every \(y\in Y\), there is at most one \(x\in X\) such
that \(P_{a,b}(x,y)=1\). In particular, every XOR game is a projection
game. For CHSH, \(A=B=X=Y=\{0,1\}\), the distribution \(p\) is uniform,
and \(s(a,b)=ab\).

For a positive integer \(k\), the \(k\)-fold parallel repetition
\(G^{\otimes k}\) is defined by sampling \(k\) independent question
pairs
\[
(a_1,b_1),\ldots,(a_k,b_k)
\]
according to \(p\). Alice receives \((a_1,\ldots,a_k)\), Bob receives
\((b_1,\ldots,b_k)\), and they return answer tuples
\((x_1,\ldots,x_k)\) and \((y_1,\ldots,y_k)\). They win if
\[
P_{a_i,b_i}(x_i,y_i)=1
\qquad\text{for every }i\in\{1,\ldots,k\}.
\]
Using independent product strategies in every coordinate gives
\[
\operatorname{val}(G^{\otimes k})
\geq \operatorname{val}(G)^k,
\qquad
\operatorname{val}^*(G^{\otimes k})
\geq \operatorname{val}^*(G)^k.
\]
These inequalities need not be equalities, because the players may correlate their answers across coordinates. In an entangled strategy, they may also use joint measurements depending on the entire tuple of questions.
For the entangled value of XOR
games, however, perfect parallel repetition holds~\cite{CSUU08}:
\[
\operatorname{val}^*(G^{\otimes k})
= \operatorname{val}^*(G)^k.
\]

\section{Classical parallel repetition of CHSH}

In this section, we prove the main result of this paper.

\begin{Theorem}
\label{thm:main}
\[ val(CHSH^{\otimes n}) \leq \left( \frac{1+\sqrt{5}}{4} \right)^n .\]
\end{Theorem}

The exponential rate in this bound is tight: Barak et al. \cite{B+} have shown that $val_{\infty}(CHSH) \geq \frac{1+\sqrt{5}}{4}$. 
Together with Theorem \ref{thm:main}, this determines $val_{\infty}(CHSH)$ exactly.

\subsection{Relaxations for the classical value}

In this section, we introduce two relaxations for the classical value of $G$, with the property that
$val (G^{\otimes n}) \leq val^+(G)^n \leq val^{++}(G)^n$. This reduces bounding
$val (CHSH^{\otimes n})$ to bounding $val^{++}(CHSH)$. We start by describing 
an expression for $val(G)$ in a form suitable for defining relaxations.

{\bf Exact expression for $val(G)$.} 
Let $c=(a, x)$ and $d=(b, y)$. Let $G$ be a matrix whose entries $G_{c, d}$ are equal to
$p(a, b)$ whenever $c=(a, x)$, $d=(b, y)$ and Alice and Bob win if they answer $x$ and $y$ to
questions $a$ and $b$. Otherwise, $G_{c, d}=0$. Since $G$ contains full information about 
the game, we can identify $G$ with the game (and use $G$ to denote both the game and its 
matrix).

A deterministic strategy for Alice can be described by a vector $u$ indexed by $c=(a, x)$ and defined as follows. 
$u_c=1$ if, given $a$, Alice answers $x$. $u_c=0$ otherwise. We can describe a deterministic strategy for Bob by a vector $v$ indexed by $d=(b, y)$
in a similar way. We note that $u$ and $v$ have the following properties: 
\begin{enumerate}
\item
$u_{a, x}, v_{b, y}\in\{0, 1\}$.
\item
for each $a$ (or $b$), there is exactly one $x:u_{a, x} = 1$ ($y: v_{b, y} = 1$, respectively).
\end{enumerate}
The probability of Alice and Bob winning if they play strategies $u$ and $v$ is equal to 
the inner product 
\[ (u, Gv) = \sum_{c, d} u_{c} G_{c, d} v_d.\] 
The value of the game $val(G)$ is the maximum of
$(u, Gv)$ over all $u, v$ that satisfy the requirements above.

{\bf Dinur-Steurer relaxation.}
Let $G$ be a two player projection game. 
The relaxed value $val^+(G)$ of the game \cite{DS} is defined as follows. 
We modify the optimization problem for $val(G)$ by allowing 
$u_c$ and $v_d$ to be vectors in $\mathbb{R}^m$
(for some $m$) which satisfy the following properties:  
\begin{enumerate}
\item
each coordinate of each $u_c$ or $v_d$ is nonnegative;
\item
for each $a$ (or $b$), we have $\sum_x \|u_{(a, x)}\|^2 =1$
(or $\sum_y \|v_{(b, y)}\|^2 =1$);
\item
if $(u_c)_i$ (or $(v_d)_i$) denotes the $i^{\rm th}$ coordinate of $u_c$ ($v_d$), 
then, given $(u_{a, x})_i$ and $(u_{a, x'})_i$, with the same question $a$ but different answers $x$ and $x'$, at most one of them is nonzero.
\end{enumerate}
The relaxed value $val^+(G)$ is equal to the maximum of
\[ \sum_{c, d} G_{c, d}(u_c, v_d) \]
over all $u, v$ that satisfy the requirements above.
It is easy to see that $val(G)\leq val^+(G)$ (because the numbers 
$u_c$ and $v_d$ in the definition of $val(G)$ can be interpreted as one-dimensional vectors).
Dinur and Steurer \cite{DS} have shown

\begin{Lemma}\cite{DS}
$val(G^{\otimes k}) \leq (val^+(G))^k$ for any projection game $G$. 
\end{Lemma}

{\bf Second relaxation.}
We now introduce an easier to compute relaxation $val^{++}(G)$ which satisfies
$val^+(G)\leq val^{++}(G)$.

$val^{++}(G)$ is defined by replacing the requirement that $\sum_x \|u_{(a, x)}\|^2 =1$
($\sum_y \|v_{(b, y)}\|^2 =1$) with a similar requirement about averages of these quantities.
That is, we require that 
\begin{equation}
\label{eq:req} 
\frac{ E_a \left[\sum_x \|u_{(a, x)}\|^2\right]+ 
E_b \left[\sum_y \|v_{(b, y)}\|^2\right]}{2} = 1 
\end{equation}
with the expectations $E_a$, $E_b$ taken over uniformly random choice of $a$ and $b$. 

We claim that restricting the vectors \(u_c\) and \(v_d\) to be
one-dimensional does not change \(\operatorname{val}_{++}(G)\).
Fix any feasible solution consisting of vectors in \(\mathbb{R}^m\),
and denote its objective value by \(V\). For each coordinate
\(i\in\{1,\ldots,m\}\), define
\[
A_i
=
\frac{
\mathbb{E}_a\left[\sum_x (u_{(a,x)})_i^2\right]
+
\mathbb{E}_b\left[\sum_y (v_{(b,y)})_i^2\right]
}{2}
\]
and
\[
B_i
=
\sum_{c,d}G_{c,d}(u_c)_i(v_d)_i.
\]
The normalization condition implies that
\[
\sum_i A_i=1,
\]
while the objective value can be written as
\[
V=\sum_i B_i.
\]
If \(A_i=0\), then all entries of the \(i\)-th coordinate that occur
in the definition of \(A_i\) vanish, and hence \(B_i=0\). Therefore,
\[
V
=
\sum_{i:A_i>0}B_i
=
\sum_{i:A_i>0}A_i\frac{B_i}{A_i}
\leq
\max_{i:A_i>0}\frac{B_i}{A_i}.
\]
Let \(i^\ast\) be an index attaining the maximum. We define
one-dimensional vectors, that is (nonnegative scalars), by
\[
\widetilde{u}_c
=
\frac{(u_c)_{i^\ast}}{\sqrt{A_{i^\ast}}},
\qquad
\widetilde{v}_d
=
\frac{(v_d)_{i^\ast}}{\sqrt{A_{i^\ast}}}.
\]
These scalars inherit the condition that, for every question, at most
one answer has a nonzero value. Moreover,
\[
\frac{
\mathbb{E}_a\left[\sum_x \widetilde{u}_{(a,x)}^2\right]
+
\mathbb{E}_b\left[\sum_y \widetilde{v}_{(b,y)}^2\right]
}{2}
=
\frac{A_{i^\ast}}{A_{i^\ast}}
=
1,
\]
so they form a feasible solution. Its objective value is
\[
\sum_{c,d}G_{c,d}\widetilde{u}_c\widetilde{v}_d
=
\frac{B_{i^\ast}}{A_{i^\ast}}
\geq V.
\]
Thus every feasible vector solution can be replaced by a feasible
one-dimensional solution of at least the same value. Consequently,
\[
\operatorname{val}_{++}(G)
=
\max_{\{u_c\},\{v_d\}}
\sum_{c,d}G_{c,d}u_cv_d,
\]
where the maximum is over nonnegative scalars \(u_c,v_d\) satisfying
the following conditions:
\begin{enumerate}
    \item For every question \(a\), at most one of the values
    \(u_{(a,x)}\) is nonzero, and, for every question \(b\), at most
    one of the values \(v_{(b,y)}\) is nonzero.

    \item The normalization condition
    \[
    \frac{
    \mathbb{E}_a\left[\sum_x u_{(a,x)}^2\right]
    +
    \mathbb{E}_b\left[\sum_y v_{(b,y)}^2\right]
    }{2}
    =
    1
    \]
    holds.
\end{enumerate}

We have $val(G) \leq val^+(G) \leq val^{++}(G)$.

\subsection{Bounds on the value of the CHSH game}

We now proceed to proving Theorem \ref{thm:main}, by showing $val^{++}(CHSH) \leq \frac{1+\sqrt{5}}{4}$.

We now solve the optimization problem and calculate $val^{++}(CHSH)$. 
Let the nonzero entries be $u_{(0, x_0)}$, $u_{(1, x_1)}$, $v_{(0, y_0)}$, $v_{(1, y_1)}$. 
The requirement on averages then becomes
\begin{equation}
\label{eq:av} \frac{u_{(0, x_0)}^2+ u_{(1, x_1)}^2+ v_{(0, y_0)}^2+v_{(1, y_1)}^2}{4}=1 
\end{equation}
and the expression to be maximized is 
\begin{equation}
\label{eq:sum} \sum_{a, b\in\{0, 1\}} G_{(a, x_a), (b, y_b)} u_{a, x_a} v_{b, y_b} .
\end{equation}
Here, $G_{(a, x_a), (b, y_b)}$ is equal to $1/4$ if the combination of answers $x_a, y_b$ is winning and 0 otherwise.
For any of the choices of $x_0, x_1, y_0, y_1\in\{0, 1\}$, there is either one or three non-zero values $G_{(a, x_a), (b, y_b)}$.
Since all terms are non-negative, a sum with one non-zero coefficient $G_{(a, x_a), (b, y_b)}$ is upper bounded by a sum with three non-zero
coefficients $G_{(a, x_a), (b, y_b)}$ one of which is in the same place as the original non-zero coefficient.

Hence, it suffices to consider the case when three of  $G_{(a, x_a), (b, y_b)}$ are non-zero.
Because of the symmetry of expression (\ref{eq:sum}) and the condition (\ref{eq:av}), we can without loss of generality assume that 
the zero value is $G_{(1, x_1), (1, y_1)}$.

For brevity, we now denote $u_{a, x_a}$ and $v_{b, y_b}$ as simply $u_a$ and $v_b$.
Then, the task is to maximize 
\[ \frac{1}{4} \left( u_0 v_0 + u_0 v_1 + u_1 v_0 \right) \]
subject to the constraint $\frac{u_{0}^2+ u_{1}^2+ v_{0}^2+v_{1}^2}{4}=1$.

We can rescale all four variables by 1/2 (setting $\tilde{u}_i = \frac{u_i}{2}$ and similarly for $v_j$'s).
Then, the expression to be maximized becomes $\tilde{u}_0 \tilde{v}_0 + \tilde{u}_0 \tilde{v}_1 + \tilde{u}_1 \tilde{v}_0 $ and
the requirement on averages becomes $\tilde{u}_{0}^2+ \tilde{u}_{1}^2+ \tilde{v}_{0}^2+\tilde{v}_{1}^2 = 1$.

Let $t = \tilde{u}_{0}^2+ \tilde{u}_{1}^2$. Then, $\tilde{v}_{0}^2+\tilde{v}_{1}^2 = 1-t$.
We express $\tilde{u}_{0}=\sqrt{t}\cos \alpha$, $\tilde{u}_{1}=\sqrt{t}\sin \alpha$,
$\tilde{v}_{0}=\sqrt{1-t}\cos \beta$, $\tilde{v}_{1}=\sqrt{1-t}\sin \beta$.
Since the variables are non-negative, we have $\alpha, \beta\in[0, \pi/2]$.
The expression to be maximized is 
\[  \sqrt{t(1-t)} (\cos\alpha\cos\beta + \cos\alpha\sin\beta + \sin\alpha\cos\beta) .\]
Regardless of $\alpha$ and $\beta$, 
the maximum of $\sqrt{t(1-t)}$ is equal to $\frac{1}{2}$ (achieved when $t=1/2$). It remains to maximize
\[ \cos\alpha\cos\beta + \cos\alpha\sin\beta + \sin\alpha\cos\beta =
\cos\alpha\cos\beta + \sin(\alpha+\beta)\]
\[ = \frac{\cos (\alpha+\beta) + \cos (\alpha-\beta)}{2} + \sin(\alpha+\beta) .\]
If we fix the sum $\alpha+\beta$ and vary $\alpha$, the only term that changes
is $\cos (\alpha-\beta)$ which reaches maximum when $\alpha=\beta$ and $\alpha-\beta=0$.
Therefore, we can assume that $\alpha=\beta$. Then, our expression-to-be-maximized is
\begin{equation}
\label{eq:1} 
\frac{\cos (2\alpha)+1}{2} + \sin(2\alpha) .
\end{equation}
Let $\gamma$ be an angle for which $\cos\gamma=\frac{1}{\sqrt{5}}$, 
$\sin\gamma=\frac{2}{\sqrt{5}}$. Then, (\ref{eq:1}) is equal to
\[ \frac{1}{2} + \frac{\sqrt{5}}{2} ( \cos (2\alpha) \cos\gamma + 
\sin (2\alpha) \sin\gamma ) 
 = \frac{1}{2} + \frac{\sqrt{5}}{2} \cos (2\alpha-\gamma) .\]
The maximum of this expression is $\frac{1+\sqrt{5}}{2}$, achieved
when $2\alpha=\gamma$. Multiplying this with $\max_t \sqrt{t(1-t)}=\frac{1}{2}$ gives 
$val^{++}(CHSH)\leq \frac{1+\sqrt{5}}{4}$.

Therefore, we have
\[ val(CHSH^{\otimes n}) \leq val^+(CHSH)^n \leq val^{++}(CHSH)^n \leq \left(\frac{1+\sqrt{5}}{4} \right)^n ,\]
 completing the proof of Theorem \ref{thm:main}.

\section{Dinur-Steurer value vs. entangled value}

We show

\begin{Theorem}
\label{thm:2}
For any XOR game $G$, we have
$val^+(G)\leq val^{*}(G)$, with equality if and only if $val(G) = val^{*}(G)$.
\end{Theorem}

\proof 
We use the standard expression for $val^{*}(G)$ from Tsirelson's theorem.
Let $\beta^*(G) = 2 val^{*}(G) - 1$ be the maximum bias towards a correct answer that can be achieved in the game $G$.
Then, by Tsirelson's theorem \cite{T},
\begin{equation}
\label{eq:exp} 
\beta^*(G) = \max_{u'_a, v'_b} \sum_{a\in A, b\in B} M_{a, b} (u'_a, v'_b) 
\end{equation}
where
\begin{itemize}
\item
the game is described by a matrix $M$ 
defined by $M_{a, b}=p(a, b) (-1)^{s(a, b)}$ where $p(a, b)$ is the probability that Alice and Bob receive questions $a, b$ and
$s(a, b)$ is the winning condition defined in (\ref{eq:xor}).
\item
the maximum is over all choices of $u'_a$ for $a\in A$ and $v'_b$ for $b\in B$
with $\|u'_a\|=\|v'_b\|=1$.
\end{itemize}

Next, we show $val^+(G)\leq val^{*}(G)$. Let $u_c$, $v_d$ be the vectors that achieve the maximum
in the optimization problem for $val^+(G)$. 
We define $u'_a=u_{a, 0} -  u_{a, 1}$
and $v'_b=v_{b, 0} -  v_{b, 1}$. 
Since $\|u_{a, 0}\|^2+\|u_{a, 1}\|^2=1$ and, in each coordinate,
only one of two vectors is non-zero, we have $\|u'_a\|=1$.
Similarly, $\|v'_b\|=1$.

For this choice of $u'_a, v'_b$, we have
\[ \sum_{a, b} M_{a, b} (u'_a, v'_b) = \sum_{a, b, x, y} p(a, b) (-1)^{s(a, b)+x+y} (u_{a, x}, v_{b, y}) \]
\[ = 2 val^{+} (G) - \sum_{a, b, x, y} p(a, b) (u_{a, x}, v_{b, y}) .\]
$val^{*}(G)\geq val^+(G)$ now follows from
\begin{Lemma}
\label{lem:1}
\[ \sum_{a\in A, b\in B, x\in X, y\in Y} p(a, b) (u_{a, x}, v_{b, y}) \leq 1 .\] 
\end{Lemma}
\proof
Let $u_a=u_{a, 0}+u_{a, 1}$ and let $v_b$ be defined similarly.
Then, $\|u_a\|=\|v_b\|=1$. 
We have
\[ \sum_{a\in A, b\in B, x\in X, y\in Y} p(a, b) (u_{a, x}, v_{b, y}) =
\sum_{a\in A, b\in B} p(a, b) (u_a, v_b) \]
\[ =
1 - \frac{1}{2} \sum_{a\in A, b\in B} p(a, b) \| u_a - v_b \|^2
\leq  1. \]
\qed

If $val^+(G)= val^{*}(G)$, we must have equality in Lemma \ref{lem:1}. This means that  
$\| u_a - v_b\|= 0$ for all pairs $(a, b)$ with $p(a, b) > 0$. We now show that this implies $val(G)=val^*(G)$.

Let $H$ be a bipartite graph with the set of vertices $A\cup B$ and the set of edges $(a, b): p(a, b)>0$.
If $H$ is connected, then equality in Lemma \ref{lem:1} implies that all $u_a$ and $v_b$ are equal.

Let \(w=(w_i)_{i=1}^m\) be the vector to which all the vectors
\(u_a\) and \(v_b\) are equal. Since \(\lVert w\rVert=1\), we have
\[
\sum_{i=1}^m w_i^2=1.
\]
For every coordinate \(i\) such that \(w_i>0\), define functions
\(f_i:A\to\{0,1\}\) and \(g_i:B\to\{0,1\}\) as follows. For every
question \(a\), let \(f_i(a)\) be the unique answer satisfying
\[
\bigl(u_{a,f_i(a)}\bigr)_i=w_i.
\]
Such a unique answer exists because
\(u_{a,0}+u_{a,1}=w\) and, in each coordinate, at most one of
\(u_{a,0}\) and \(u_{a,1}\) is nonzero. Similarly, for every question
\(b\), let \(g_i(b)\) be the unique answer satisfying
\[
\bigl(v_{b,g_i(b)}\bigr)_i=w_i.
\]
Let \(S_i=(f_i,g_i)\) be the deterministic strategy in which Alice answers $x=f_i(A)$ and Bob answers $y=g_i(B)$. Then,
\[
\begin{aligned}
\operatorname{val}_{+}(G)
&=
\sum_{a,b,x,y}
G_{(a,x),(b,y)}
\left\langle u_{a,x},v_{b,y}\right\rangle \\
&=
\sum_{i:w_i>0}w_i^2
\sum_{a,b}
G_{(a,f_i(a)),(b,g_i(b))} \\
&=
\sum_{i:w_i>0}w_i^2
\Pr\bigl(S_i\text{ wins}\bigr).
\end{aligned}
\]
Since \(S_i\) is a deterministic strategy,
\[
\Pr\bigl(S_i\text{ wins}\bigr)\leq \operatorname{val}(G).
\]
Consequently,
\[
\operatorname{val}_{+}(G)
\leq
\sum_{i:w_i>0}w_i^2\operatorname{val}(G)
=
\operatorname{val}(G).
\]
The reverse inequality
\(\operatorname{val}(G)\leq\operatorname{val}_{+}(G)\) holds by
definition, and therefore
\[
\operatorname{val}(G)=\operatorname{val}_{+}(G)
=\operatorname{val}^{*}(G),
\]
as required.

If $H$ is not connected, we can apply the same argument to each connected component. The classical, relaxed and entangled values 
decompose as weighted averages over connected components.
\qed

From Theorem \ref{thm:2}, it follows that

\begin{Corollary}
$val(G)<val^*(G)$ implies $val_{\infty}(G)<val^*(G)$.
\end{Corollary}

\proof
From Theorem \ref{thm:2}, $val(G)<val^*(G)$ implies $val^+(G)<val^*(G)$. 
The results of Dinur and Steurer \cite{DS} imply that $val_{\infty}(G)\leq val^+(G)$.
\qed
\section{Conclusion}

We have determined the asymptotic classical value of the parallel-repeated
CHSH game. Using the multiplicative relaxation
\(\operatorname{val}_{+}\) of \cite{DS}, we proved that, for every
\(n\geq 1\),
\[
\operatorname{val}(\mathrm{CHSH}^{\otimes n})
\leq
\left(\frac{1+\sqrt{5}}{4}\right)^n.
\]
Together with the matching lower bound of Barak et al., this gives
\[
\operatorname{val}_{\infty}(\mathrm{CHSH})
=
\frac{1+\sqrt{5}}{4}.
\]

We also showed that every XOR game \(G\) satisfies
\[
\operatorname{val}_{+}(G)\leq\operatorname{val}^{*}(G),
\]
with equality if and only if
\[
\operatorname{val}(G)=\operatorname{val}^{*}(G).
\]
Thus, every quantum-classical gap for an XOR game remains strict at the
level of the corresponding asymptotic values for the corresponding repeated game.

A natural question is whether this qualitative statement can be strengthened quantitatively. More precisely, is there a universal
constant \(c>0\) such that every XOR game satisfies
\[
\operatorname{val}_{+}(G)
\leq
c\,\operatorname{val}(G)
+
(1-c)\operatorname{val}^{*}(G)?
\]
For example, \(c=1/2\) would imply that at least one half of the
single-copy quantum-classical gap is preserved between the asymptotic
classical and quantum winning rates.


\begin{thebibliography}{9}
\bibitem{A}
S. Aaronson.  The NEW Ten Most Annoying Questions in Quantum Computing.
Blog post, May 13, 2014, https://scottaaronson.blog/?p=1792.

\bibitem{Araujo}
M. Araújo, F. Hirsch, M. Quintino. Bell nonlocality with a single shot. 
{\em Quantum} 4:353, 2020.

\bibitem{B+}
B. Barak, M. Hardt, I. Haviv, A. Rao, O. Regev, D. Steurer.
Rounding parallel repetitions of unique games.
{\em Proceedings of FOCS'2008}, pp. 374-383.

\bibitem{T}
B. Cirel'son (Tsirelson). 
Quantum generalizations of Bell's inequality. {\em Letters in Mathematical
Physics}, 4:93-100, 1980.

\bibitem{CHSH69}
J.~F. Clauser, M.~A. Horne, A. Shimony, and R.~A. Holt,
Proposed experiment to test local hidden-variable theories,
\emph{Physical Review Letters}, 23:880-884, 1969.

\bibitem{Cleve}
R. Cleve, P. H\o yer, B. Toner, J. Watrous:
Consequences and Limits of Nonlocal Strategies. 
{\em Proceedings of CCC'2004}, pp. 236-249.

\bibitem{CSUU08}
R. Cleve, W. Slofstra, F. Unger, and S. Upadhyay,
Perfect parallel repetition theorem for quantum XOR proof systems,
\emph{Computational Complexity}, 17:282-299, 2008.

\bibitem{DS}
I. Dinur, D. Steurer. 
Analytical approach to parallel repetition.
{\em Proceedings of STOC'2014}, pp. 624-633.

\bibitem{DSV15}
I. Dinur, D. Steurer, and T. Vidick,
A parallel repetition theorem for entangled projection games,
\emph{Computational Complexity}, 24:201-254, 2015.

\bibitem{Hol09}
T. Holenstein,
Parallel repetition: Simplification and the no-signaling case,
\emph{Theory of Computing}, 5:141-172, 2009.

\bibitem{Raz98}
R. Raz,
A parallel repetition theorem,
\emph{SIAM Journal on Computing}, 27:763-803, 1998.


\end{thebibliography}
\end{document}